\documentclass[a4paper,UKenglish,cleveref, autoref, thm-restate, nolineno]{socg-lipics-v2021}
\usepackage{hyperref,url}
\hideLIPIcs  %uncomment to remove references to LIPIcs series (logo, DOI, ...), e.g. when preparing a pre-final version to be uploaded to arXiv or another public repository

\makeatletter
\def\s@linecount{} % Vide la macro d'affichage du nombre total de lignes
\makeatother

\AtBeginDocument{\nolinenumbers}

\title{\textsc{HyperLogLog} for probabilists} %TODO Please add
\author{Lucas Gerin}{Universit\'e Paris Nanterre, Bâtiment Allais, 200 avenue de la République, 92000 Nanterre (France)  }{lucas.gerin@parisnanterre.fr}{}{This work is supported by \textsc{Anr} Project \textsc{Louccoum}.}%TODO mandatory, please use full name; only 1 author per \author macro; first two parameters are mandatory, other parameters can be empty. Please provide at least the name of the affiliation and the country. The full address is optional. Use additional curly braces to indicate the correct name splitting when the last name consists of multiple name parts.

\authorrunning{L.Gerin} %TODO mandatory. First: Use abbreviated first/middle names. Second (only in severe cases): Use first author plus 'et al.'

\Copyright{L.Gerin} %TODO mandatory, please use full first names. LIPIcs license is "CC-BY";  http://creativecommons.org/licenses/by/3.0/

 %TODO mandatory; please add comma-separated list of keywords

\category{} %optional, e.g. invited paper

\relatedversion{} %optional, e.g. full version hosted on arXiv, HAL, or other respository/website
\acknowledgements{I want to warmly thank \'Eric Fusy for our discussions about several aspects of the original \texttt{HyperLogLog} article, and possible improvements. Thanks also to Seth Pettie for pointing out \cite{wang2023better}.}%optional

\AtBeginDocument{
  \definecolor{lipicsYellow}{HTML}{FF00FF} % Code couleur Fuchsia (ou en RGB : {1.0, 0.0, 1.0})
}

\AtBeginDocument{
  \DeclareCaptionLabelFormat{lipics}{\colorbox{lipicsYellow}{\bfseries\color{white}#1~#2}}
}
\usepackage{algorithm}
\usepackage[noend]{algpseudocode}
\makeatletter
\def\BState{\State\hskip-\ALG@thistlm}
\makeatother

\newtheorem{hypothesis}{Hypothesis}

\usepackage{multicol}
\usepackage{pgfplots}  % Pour le TikZ
\pgfplotsset{compat=1.17}

\newcommand{\eps}{\varepsilon}

\begin{document}

% Ou soumettre
% PAYAAAAANT -  ACM Transactions Algorithms - Woodruff, Conrado Martinez
% Probability in the Engineering and Informational Sciences - bof
% Discrete Mathematics & Theoretical Computer Science
% PAYAAAAANT - Algorithmica - Woodruff
% Information Processing Letters
% Journal of Applied Probability - bof
% Journal of Machine Learning Research

\maketitle
\nolinenumbers

\begin{abstract}
\texttt{HyperLogLog} is a now classic probabilistic algorithm that provides an approximation of the number of distinct elements in a massive dataset, using only one pass over the data. In the original article, Flajolet, Fusy, Gandouet, Meunier (2007) provided a sharp analysis of the expectation and variance of the output, using explicit formulas analyzed using poissonization and Mellin transform.

In this short article, we revisit the analysis of \texttt{HyperLogLog} with a more probabilistic viewpoint. This allows us to establish exponential deviation inequalities for the \texttt{HyperLogLog} estimator. The methods are elementary, but the estimates are non-asymptotic and totally explicit. 
\end{abstract}

\section{Introduction: context and main results}
\subsection{A brief overview of \texttt{HyperLogLog}}
\label{sec:typesetting-summary}

This work is devoted to the analysis of  \texttt{HyperLogLog}, a now classic probabilistic  \emph{streaming algorithm}. Streaming algorithms allow only one pass over the data: each input is observed, quickly processed to update internal variables and then discarded. This is a severe constraint 
and one accepts that the output is an approximation of the actual solution, provided by a summary (or \emph{sketch}) of the data stream. Most efficient data stream algorithms use randomness or pseudo-randomness in order to overcome the constraints. A concise introduction to data stream algorithms can be found in \cite[Chap.7]{erickson2023algorithms}.

An iconic example of an algorithmic problem solved by data stream algorithms is that of \emph{cardinality estimation}. The input is a multiset $\mathcal{M}$ of elements of a large set $\mathcal{D}$, and the output is an estimation of the unknown number $N$ of distinct elements in $\mathcal{M}$.
Many probabilistic data stream algorithms have been introduced for cardinality estimation; we refer to \cite{heule2013hyperloglog,harmouch2017cardinality,pettie2021information} for comparisons of several algorithms and numerical experiments.% most of them fit into two families:
%\begin{itemize}
%\item \emph{Sketch-based} algorithms which scan the entire dataset, use a hash function to obtain pseudo-random uniform and record a small \emph{sketch} of the data allowing to estimate $N$. This sketch can be based on order statistics \cite{bar2002counting,giroire2009order} or bit patterns  \cite{alon1996space,durand2003loglog,flajolet2007hyperloglog}.
%\item Algorithms using random sampling of the dataset \cite{whang1990linear}  or random projections  
%\cite{kane2010optimal}.
%\end{itemize}

Among them \texttt{HyperLogLog} is a very popular algorithm introduced by Flajolet, Fusy, Gandouet, and Meunier in \cite{flajolet2007hyperloglog} to solve the cardinality estimation problem. It is easy to implement and very effective \cite{heule2013hyperloglog}, it is implemented in many modern database systems such as \texttt{Redis} \cite{Redis} or \texttt{BigQuery} \cite{BigQuery}. The history of \texttt{HyperLogLog}, put in the context of probabilistic streaming algorithms, is nicely told in \cite{flajolet2004counting,lumbroso2018story}.
%The general principle behind \emph{min-sketch} algorithms is:
%\begin{enumerate}
%\item to hash the elements of $\mathcal{M}$ using a hash function to obtain pseudo-random uniform numbers in $[0,1]$
%\item use the statistical properties of the extremal values in a independent sample of size $N$ in order to estimate the size of $\mathcal{M}$. 
% \end{enumerate}

The original algorithm is presented in Algorithm \ref{HLL} just below, let us briefly give the general scheme. The idea, which dates back to a pioneering work by Flajolet and Martin \cite{flajolet1985probabilistic}, is to analyze certain patterns appearing in the pseudo-random binary strings $\{h(v), v\in \mathcal{M}\}$ for some hash function $h$. Let $\rho(s)\in \mathbb{Z}_{\geq 1}$ denote the position of the leftmost $1$ in the binary string $s$ (\emph{e.g.} $\rho(001011\dots)=3$). If $h$ is well designed then each pseudo-random integer $\rho(h(v))$ should resemble a geometric random variable of mean $2$. Therefore it is expected that
$$
\mathbb{P}\left(\max_{v\in \mathcal{M}} \{\rho(h(v))\} \leq k \right)\approx (1-2^{-k})^N \approx \exp(-N2^{-k}),
$$
so that the maximal value  $R:= \max_{v\in \mathcal{M}} \{\rho(h(v))\}$ is typically of order $\log_2(N)$. Then the pseudo-random variable $2^{R}$ (or a variant of it) could be used to estimate $N$. To reduce the variance, an important principle is to parallelize this estimation by searching for patterns in different locations of the strings. In practice one uses $m$ parallel registers which allow to emulate $m$ \emph{almost} independent experiments while requiring only a single hash function. This technique, called \emph{stochastic averaging}, was already used in \cite{flajolet1985probabilistic}. 

%Quoting from \cite{haas1995sampling}: "Virtually all query optimization methods in relational and
%object-relational database systems require a means of assessing the number of distinct values of an attribute in a
%relation".

\medskip

In order to formally describe the algorithm we need to choose:
\begin{itemize}
\item A \emph{hash function} $h:\mathcal{D} \to [0, 1] $; 
\item A parameter $b\in \mathbb{Z}_{\geq 0}$, and a memory composed of $m=2^b$ registers. Each register is a memory unit in which one can store a float or a large integer. 
\end{itemize}
(Early experiments \cite{flajolet2007hyperloglog,heule2013hyperloglog} suggested to use a $32$-bit or $64$-bit hash function and $m\in \{64,128,256\}$.)

\medskip

The algorithm runs as follows (see a sketch in \cref{fig:SchemaHLL} to summarize notation):
\begin{algorithm}
\caption{HyperLogLog}\label{HLL}
\begin{algorithmic}[1]
 \State\label{l:input}  \textbf{input:} multiset $\mathcal{M}$ of items from domain $\mathcal{D}$
 \State\label{l:par}  \textbf{parameters:} integer $b\geq 1$; $m \gets 2^b$
% \State\label{l:phantom}  $\phantom{\text{\textbf{parameters:}}}$ hash function  $h$
 \State\label{l:precomp}  \textbf{precomputation:} constant $\alpha_m$
\State\label{l:init} \textbf{initialize} a collection of $m$ registers $R[1]=\dots =R[m]=1$
\For{$v\in \mathcal{M}$} 
\State\label{l:x} $x \gets h(v)$ and write $x=0.x_1x_2x_3\dots$ \hfill // hashing
\State\label{l:j} $j \gets 1+\langle x_1x_2\dots x_b\rangle_2$ 	\hfill // register $1\leq j\leq m$ given by the first $b$ bits of $x$
\State\label{l:w} $w \gets 0.x_{b+1}x_{b+2}\dots $ \hfill  // extracting a fresh binary string out of $x$
\State\label{l:M} $R[j] \gets \mathrm{max}\{R[j],\rho(w)\}$ \hfill // updating the $j$-th register 
\EndFor
\State\label{l:Z} $Z_N \gets \left(\sum_{j=1}^m 2^{-R[j]}\right)^{-1}$
\State\label{l:hat_N} \textbf{return} $\hat{N}:=\alpha_m m^2\times Z_N$  \hfill  // debiasing 
\end{algorithmic}
\end{algorithm}
\begin{remark}
\begin{itemize}
\item The constant $\alpha_m$ at Line \ref{l:precomp} has an explicit integral form, and may be computed with arbitrary precision (see   \cite[Section 3.2]{flajolet2007hyperloglog}). In particular it is shown that the sequence $(\alpha_m)$ is bounded and $\alpha_m\stackrel{m\to+\infty}{\to} (2\log(2))^{-1}=0.72135\dots$. (Algorithm \ref{HLL} can actually be implemented with any integral value of $m$, not necessarily a power of two.)
\item In \cite{flajolet2007hyperloglog}, two variants are proposed for the initialization of $R[j]$'s at Line \ref{l:init}: $R[j]\leftarrow 0$ or $R[j]\leftarrow -\infty$. Here we make a slight modification and initialize $R[j]\leftarrow 1$ in order to simplify the analysis. This does not change the asymptotic behavior provided that $N \gg m\log(m)$, %(see \cref{sec:init})
 which is of course the case in practice.
\end{itemize}
\end{remark}

\begin{figure}
\begin{center}
\includegraphics[width=15cm]{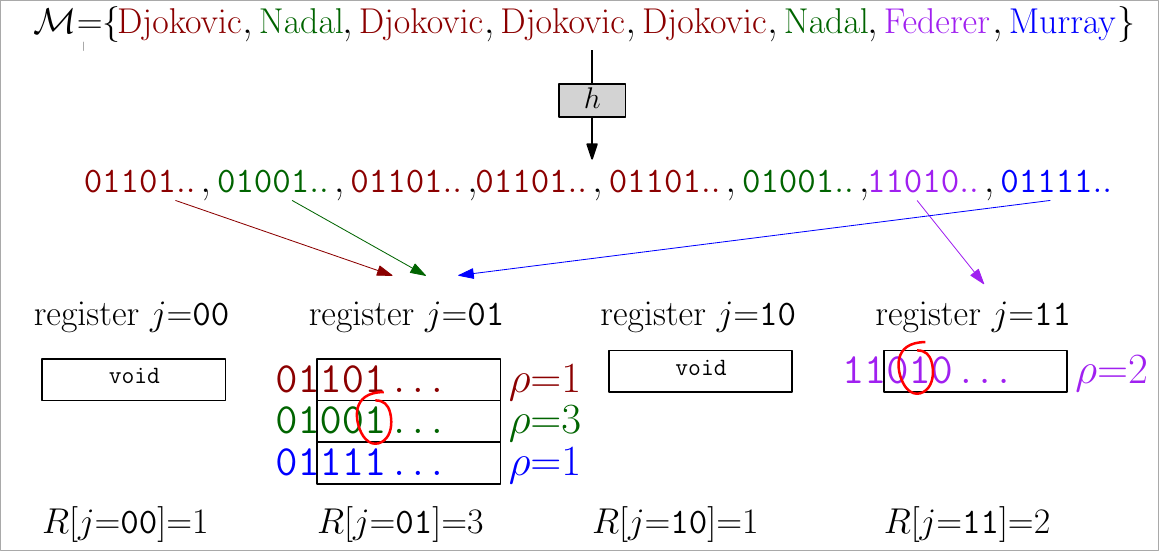}
\caption{A run of \cref{HLL} for $b=2$ (\emph{i.e.} $m=2^2$) used to estimate the number of distinct Grand Slam winners in 2011-12. For instance, $\text{Djokovic}$, $\text{Nadal}$ and $\text{Murray}$ are all sent onto the register $\mathtt{01}$ since the corresponding hashed values all three begin with $\mathtt{01}$.  The maximal leftmost $1$  in each register is surrounded in red, illustrating that in this run of \texttt{HyperLogLog} we obtain $m^2Z_N=4^4\left(2^{-1}+2^{-3}+2^{-1}+2^{-2}\right)^{-1}$.}
\label{fig:SchemaHLL}
\end{center}
\end{figure}

At this stage it is quite difficult to guess why $\hat{N}$ would be a reasonable estimation of $N$, % (apart from the case $m=1$, compare with \cref{eq:prehistoire}), 
we give an intuition below in \cref{sec:heuristic}. 
For the analysis, the authors of  \cite{flajolet2007hyperloglog} assume that the hash function $h$ is well designed. More formally they make the following mathematical assumption, very common in the literature  (see the discussion in \cite{pettie2021information}) and referred to as the \textsc{RandomOracle} model.
\begin{hypothesis}[\textsc{RandomOracle} model]
\label{hyp:IM}
Let $N$ be the number of distinct elements in $\mathcal{M}$. The set of hashed values $\{h(v), v\in \mathcal{M} \}$ (where duplicates have been removed) has the same statistical properties as a set of $N$ independent random variables uniform in $(0,1)$.
\end{hypothesis}

Throughout this paper we assume as in \cite{flajolet2007hyperloglog} that \cref{hyp:IM} holds.
For implementation we  of course have to take into account the fact that hashed values $h(v)$ are stored only up to  a finite precision. In \cite[Sec.4]{flajolet2007hyperloglog} the authors propose for large $N$ a correction which takes care of collisions between hashed values.

Here is the main result of the original \texttt{HyperLogLog} article.%\cite{flajolet2007hyperloglog}.
\begin{theorem}[Flajolet-Fusy-Gandouet-Meunier, Theorem 1 in \cite{flajolet2007hyperloglog}]\label{th:historique}\ \\
Assume \cref{hyp:IM}.\\
For all $m\geq 3 $, Algorithm \ref{HLL} is asymptotically almost unbiased. More precisely, when $N\to+\infty$,
\begin{align}
\mathbb{E}\left[\hat{N}\right]&= N(1+\delta_1(N)+\mathrm{o}(1)),\\
\sqrt{\mathrm{Variance}\left(\hat{N}\right)}&= N\left(\frac{1}{\sqrt{m}}\beta_m + \delta_2(N) + \mathrm{o}(1)\right).\label{eq:VarianceHLL}
\end{align}
where
\begin{itemize}
\item $(\beta_m)_{m\geq 3}$ is an explicit sequence converging to $\beta_\infty := \sqrt{3\log(2)-1}<1.04$,
\item $\delta_1(N),\delta_2(N)$ are oscillating functions of a tiny amplitude satisfying, whenever $m\geq 2^4$,
$$
\limsup_{N\to+\infty} |\delta_1(N)|< 5.10^{-5},\qquad \limsup_{N\to+\infty} |\delta_2(N)|< 5.10^{-4}.
$$
\end{itemize}
\end{theorem}

\begin{remark}
Let us make a few comments regarding the above result.
\begin{itemize}
\item If $m=1$ then none of the corresponding constants $\alpha_1$ nor $\beta_1$ is finite. In this case we have $Z_N=2^{R[1]}$ which has an infinite expectation under \cref{hyp:IM} (see \cref{rem:m_1}). If $m=2$ then $\alpha_2<+\infty$ but $\beta_2=+\infty$  \cite[Sec.3.2]{flajolet2007hyperloglog}.
\item The oscillation phenomenon is due to the purely discrete nature of the algorithm and is typical in the analysis of extreme values of geometric random variables (see \emph{e.g.} \cite{brands1994number}).
\item For \cref{th:historique} to be correct the registers of Algorithm \ref{HLL} need to store arbitrarily large integers $R[1],\dots,R[m]$. 
Under the assumption that $N\leq N_{\mathrm{max}}$ for some $N_{\mathrm{max}}$ the space complexity may however be restricted to $\mathcal{O}(\log\log(N_{\mathrm{max}}))$  and still returns an almost unbiased estimation of $N$  (see \cite[p.152]{flajolet2007hyperloglog} or \cite[Sec.4]{durand2003loglog}). %, or the discussion below in \cref{sec:space}) 
Hence the name of the algorithm. 
\end{itemize}
\end{remark}

%, which is quite common when analyzing the maxima of geometric random variables (see for example \cite{eisenberg1993asymptotic}).

A natural question that motivated this work is to go beyond the first two moments and study the deviations of the random variable $\hat{N}$ (or $m^2Z_n$). Quoting the authors of \texttt{LogLog}  (the older sister of \texttt{HyperLogLog}) in \cite{durand2003loglog}:
\begin{center}
\begin{minipage}{0.9\textwidth}
\emph{Consequently, the estimate  returned is very roughly Gaussian: at any rate, it has exponentially decaying tails. (In principle, a full analysis would be feasible.)}
\end{minipage}
\end{center}
%In order to supplement the estimate provided by \cref{th:historique}, we therefore seek to establish non-asymptotic deviation inequalities for $\hat{N}$ .
%\medskip
%For practical applications, it is desirable to have the best possible deviation inequalities, with explicit constants.
%
%\medskip
Another motivation came from a master’s-level course on probabilistic algorithm analysis, for which we wanted to present the principles of \texttt{HyperLogLog}. 
Unfortunately, the analysis of $\mathbb{E}[Z_N],\mathbb{E}[Z_N^2]$ are quite intricate and it is difficult to grasp the intuition of the factor $\alpha_m m^2$ in the expression of $\hat{N}$ (even just regarding the factor $m^2$).
The proof of \cref{th:historique} relies on subtle estimates using (i) an explicit (though rather complex) formula for $\mathbb{E}[\hat{N}]$ (ii) poissonization (iii) asymptotic properties of the Mellin transform. %(The Mellin transform appears to be particularly well adapted to handle the phenomenon of tiny oscillations.) 

 We hope that a more probabilistic approach addresses this issue, if only for educational purposes. 
 The results below are based solely on the Chernov method and simple concavity arguments.
 Furthermore, this work may be adaptable to other algorithms based on \emph{stochastic averaging}. We illustrate this in \cref{sec:MinCount} using \texttt{MinCount}, a continuous analogue of \texttt{HyperLogLog}.
 
%%%%%%%%%%%%%%%%%%%%%%%%%%%%%%%%%%%%
\subsection{Main result and comments}
Theoretical lower bounds for the cardinality estimation (\cite[Prop.4.1]{alon1996space}, \cite{indyk2003tight,pettie2021information}) essentially tell us that for every probabilistic data stream algorithm with output $\tilde{N}$ we have that $\mathbb{P}\left(\big|\tilde{N} -N\big|\geq \eps N\right)$ is bounded away from zero when $N\to +\infty$.

Regarding \texttt{HyperLogLog}, the Bienaym\'e-Chebyshev inequality combined with \cref{th:historique} gives the following upper bound. %\footnote{Here for simplicity we  ignore the fact that  \cref{eq:VarianceHLL} is only asymptotic.}. 
For every $\eps>0$ 
\begin{equation}\label{eq:BT}
\mathbb{P}\left(\big|\hat{N} -N\big|\geq \eps N\right)\leq
\frac{\mathrm{Var}(\hat{N})}{\eps^2N^2}\stackrel{N\text{ large}}{\approx} 
 \frac{\beta_m^2}{m}\frac{1}{\eps^2}.
\end{equation}
%(Of course the above display is relevant only if the RHS is less than $1$. To give an idea this requires that $\eps>0.2559\dots$ if $m=256$.) 
Even though the final values $R[1],\dots, R[m]$ are not truly independent, the intuition behind \emph{stochastic averaging} suggests that \eqref{eq:BT} could be improved into an upper bound exponential in $m$.
On the other hand, since $\mathrm{Variance}(\hat{N})$  is of order of $N^2$, one cannot be too demanding for small $\eps$. Simulations show for instance that for $\eps=0.1, m=128$ and $N=10^9$:
$$
\begin{array}{r l l}
&\mathbb{P}\left(\hat{N} \leq 0.9N\right) \approx 13\%, & \mathbb{P}\left(\hat{N} \geq 1.1N\right) \approx 14\%.\\
%m=256: &\mathbb{P}\left(\hat{N} \leq 0.9N\right) \approx 5\%, & \mathbb{P}\left(\hat{N} \geq 1.1N\right) \approx 7\%.
\end{array}
$$
We now state our main result. To this purpose we introduce for $x>0$ the \emph{rate} function
\begin{minipage}[c][6cm][c]{0.48\textwidth}
\centering
\begin{tikzpicture}
\begin{axis}[
    axis lines = middle,
    xlabel = $x$,
    ylabel = {$\mathcal{J}(x)$},
    domain=0.01:2,
    samples=200,
    ymin=-0.5, ymax=2,
    xmin=0, xmax=2,
    grid=both,
    width=6cm
]
\addplot[blue, thick] {1/x - 1 + ln(x)};
\end{axis}
\end{tikzpicture}
\end{minipage}
\hfill
\begin{minipage}[c][6cm][c]{0.48\textwidth}
\begin{equation}\label{eq:Rate}
\mathcal{J}(x)=1/x-1+\log(x).%\tag{\text{\ref{eq:main-left}}}
%\mathcal{I}(x)=x-1-\log(x).%\tag{\text{\ref{eq:main-left}}}
\end{equation}
\end{minipage}
Note that $\mathcal{J}(x)> 0$ for all $x\neq 1$.

\begin{theorem}\label{th:main}\ \\
Assume \cref{hyp:IM}. For all $b\geq 0$, $m=2^b$ the output of \emph{\texttt{HyperLogLog}} satisfies the following estimates for every $N\geq 1$.
\begin{itemize}
\item {\bf Left-tail.} For all $\mu \leq 1$
\begin{equation}\label{eq:main-left}
\mathbb{P}\left(m^2 Z_N \leq \mu N\right)\leq  \exp\left(-m\mathcal{J}(\mu) \right).
\end{equation}
\item {\bf Right-tail.} For all $\lambda \geq 2$
\begin{equation}\label{eq:main-right}
\mathbb{P}\left(m^2 Z_N \geq \lambda N\right)\leq \exp\left(-m\mathcal{J}(\lambda/2 ) +\delta(m,N,2/\lambda) \right),
\end{equation}
where $ \delta(m,N,c)$ is an explicit function given in \cref{eq:delta} and satisfies $ \delta(m,N,c)\leq (c+e^{-2})m^2/N$ for all $m,N \geq 1$, $ c\in(0,1)$.
\end{itemize}
\end{theorem}
The proofs of both bounds rely on the Chernov method applied to some random variables $Y_1,\dots,Y_m$ related to $R[1],\dots, R[m]$ and introduced below in \cref{prop:LoiYj}. Of course the difficulty is that $R[1],\dots, R[m]$ are not independent.
\begin{remark} Several comments are in order.
\begin{enumerate}
\item We emphasize that \cref{eq:main-left,eq:main-right} are non-asymptotic: they hold for every $N\geq 1$. Observe however that these bounds are  meaningless for very small $N$. In the degenerate cases where $N\ll m$ then the LHS of \cref{eq:main-left} is zero while the RHS of \cref{eq:main-right} is larger than $1$.
\item 
Both bounds \cref{eq:main-right,eq:main-left} are similar as they are exponentially small in $m$ and involve the rate function $\mathcal{J}$. However we observe two important differences between the two expressions:
\begin{itemize} 
\item For the right-tail we need the correction $\delta$ for small values of $N$. To understand why this is necessary, consider the extreme case where $N=o(m)$. Then at most $N=o(m)$ registers $R[j]$'s have been updated at the end of the algorithm and $Z_N=\frac{1}{m-o(m)}$. Hence 
$m^2Z_N =\Omega(m)$ and the event $\{m^2 Z_N \geq \lambda N\}$ which arises on the LHS of \cref{eq:main-right} occurs certainly. %On the opposite, the event $\{m^2 Z_N \leq \mu N\}$ on the LHS of  \cref{eq:main-left} occurs with a very small probability, even for small values of $N$.\\
Hence, a concentration inequality like \cref{eq:main-right} cannot hold without a correction term for small $N$'s.
\item Unfortunately, the right-tail inequality holds only for $\lambda \geq 2$. This is due to a necessary manipulation to eliminate the integer parts (see \cref{eq:EncadrementZ} below). We discuss in \cref{sec:vraie_laplace}  how one could lower the constant.
\end{itemize}
\item %We also point out that the statement holds even for $m\in\{1,2\}$. 
Even if the analysis of the original \texttt{HyperLogLog} requires $m\geq 3$,  \cref{th:main}  shows that it still makes sense to estimate $N$ with $m^2Z_N$ in the extreme cases $m=1$, $m=2$. To give an idea one gets for $m=2$ and $N=10^9$ that 
$$
\mathbb{P}\left( 2^2Z_N \leq \frac{1}{10} N\right)<  1.53\times 10^{-6},\qquad
\mathbb{P}\left( 2^2Z_N \geq 10 N\right) <  0.1982,
$$
so $2^2Z_N$ is a reasonable estimation of the order of magnitude of $N$.%We claimed that $\mathbb{E}[Z_N]=+\infty$ if $m=1$, which may seem contradictory at first glance with \cref{eq:main-right} as it appears to suggest that $Z_N$ has an exponential tail. However, 
\item \cref{eq:main-right}  implies that for every fixed $m,N$ and a large $t>0$, 
$$
\mathbb{P}\left(\hat{N} \geq t\right)\leq \frac{C_{m,N}}{t^{m}}%C_{m,N} \left(\frac{2N}{t}\right)^m
$$
for some constant $C_{m,N}>0$ not depending on $t$. This shows that $\hat{N}$ has a finite $p$-th moment for all $p<m$. %(Compare with \cref{th:historique}: $\mathbb{E}[Z_N]=+\infty$ if $m=1$ and $\mathbb{E}[Z_N^2]=+\infty$ if $m=2$.)
\end{enumerate}
\end{remark}

Let us now compare \cref{th:main} with \cref{eq:BT} provided by \cite{flajolet2007hyperloglog}. 
If we apply \cref{th:main} we obtain, recalling that $\hat{N}=\alpha_m m^2Z_N$:
\begin{itemize}
\item For $\eps >1-\alpha_m$ ($\approx 0.27\dots $ for large $m$)
\begin{align}
\mathbb{P}\left(\hat{N} \leq(1-\eps) N\right)%&= \mathbb{P}\left(m^2Z_N \leq \frac{1-\eps}{\alpha_m} N\right)\notag\\
&\leq \exp\left(-m\mathcal{J}((1-\eps)/\alpha_m) \right).\label{eq:left_E}
\end{align}
\item For $\eps>2\alpha_m -1$ ($\approx 0.44\dots$ for large $m$)
\begin{align}
\mathbb{P}\left(\hat{N} \geq (1+\eps) N\right)%&=  \mathbb{P}\left(m^2Z_N \geq \frac{1+\eps}{\alpha_m} N\right)\notag\\
&\leq \exp\left(-m\mathcal{J}((1+\eps)/2\alpha_m) +\delta(m,N,2\alpha_m/(1+\eps)) \right).\label{eq:right_E}
\end{align}
\end{itemize}
To give an idea, here is  the result for  $m=128$ and $N=10^9$ if we want to estimate $N$ to within a factor $2$ (above or below):
\begin{center}
\begin{tabular}{|c|c | c |}
\hline
& eq.\eqref{eq:BT}: Bienaym\'e-Chebyshev  + \cite{flajolet2007hyperloglog}& eq.\eqref{eq:right_E} and \eqref{eq:left_E} \\
\hline \hline
$\mathbb{P}\left(\hat{N} \geq 2N\right)$ & $ \leq 0.0086$ &  $\leq 0.0016$\\
$\mathbb{P}\left(\hat{N} \leq N/2\right)$ & $ \leq 0.0342$ &  $\leq 9.35\times 10^{-5}$\\
\hline
\end{tabular}
\end{center}

\subsection{Some related probabilistic results}

When a probabilistic data stream algorithm can be executed in parallel and independently (this is for example the case of an important algorithm by Kane, Nelson and  Woodruﬀ \cite{kane2010optimal}), it is clear that the deviation probabilities decrease exponentially. For algorithms based on \emph{stochastic averaging} the analysis is more delicate because of dependence, we mention some results related to the probabilistic analysis of \texttt{HyperLogLog} and some of its variants.

Beyond the estimation of the two first moments, the analysis provided in \cite{flajolet2007hyperloglog} already suggested that for large $N$ the distribution of $\hat{N}/N$ is close to a Gamma random variable, paving the way for the study of exponential deviations. Shortly after, Lumbroso \cite{lumbroso2010optimal} indeed proposed a continuous variant of \texttt{HyperLogLog} and proved the convergence of the renormalized output to a Gamma random variable. {\L}ukasiewicz and Uzna{\'n}ski \cite{lukasiewicz2022cardinality} proposed another continuous variant of \texttt{HyperLogLog} in which the sketch of the data takes the form of independent Gumbel random variables, thereby greatly simplifying the probabilistic analysis.

We also mention a related paper by Clifford and Cosma \cite{clifford2012statistical} who analyzed
several families of probabilistic data stream algorithms with  the viewpoint of asymptotic statistics and information theory. Besides, a fairly detailed probabilistic analysis of the correcting terms in \texttt{HyperLogLog} for large $\hat{N}$ is provided in \cite{ertl2017new}.

Wang and Pettie \cite{wang2023better} have proposed a very powerful general theoretical framework that encompasses several min-sketch-type algorithms, including \texttt{HyperLogLog}. They analyze in particular a one-parameter deformation of the estimator for some real parameter $\tau$, which they call $\tau$-GRA estimator (the classical \texttt{HyperLogLog} corresponds to $\tau=1$). They show that, by choosing the parameter carefully, the variance of $\hat{N}$ can be further reduced (see \cite[Th.3]{wang2023better}, note that for technical reasons they consider a smoothed model.). Up to adjusting the constants,  \cref{th:main} also applies to the $\tau$-GRA estimator.

%%%%%%%%%%%%%%%%%%%%%%%%%%%%%%%%%%%%
\section{Preliminaries and heuristic}
\label{sec:heuristic}

In this section, we introduce some useful notation and observations for the proof of the main theorem. We first state the following consequence of \cref{hyp:IM} (it is essentially equivalent to \cite[Prop.1]{flajolet2007hyperloglog}).
\begin{proposition}\label{prop:LoiYj}
Assume \cref{hyp:IM}. \\
Let $N\geq 1, b\geq 1, m=2^b$. Let $(A_1,\dots, A_m)$ be distributed as a multinomial distribution with parameters $(N; \tfrac1m,\dots \tfrac1m)$. Conditionally to $A_1=a_1,\dots, A_m=a_m$ let $Y_1,\dots ,Y_m$ be independent and drawn as follows:
\begin{itemize}
\item if $a_j=0$ then $Y_j=1$;
\item if $a_j>0$ then $Y_j$ has density $a_j(1-y)^{a_j -1}$ over the interval $(0,1)$.
\end{itemize}
Then the variable $Z_N$ defined in Line \ref{l:Z} of  Algorithm \ref{HLL} satisfies
\begin{equation}\label{eq:Z_proba}
Z_N \stackrel{\text{\emph{(d)}}}{=} \frac{1}{\sum_{j=1}^m 2^{-\rho (Y_j) }}.%=\frac{1}{\sum_{j=1}^m 2^{\lfloor \log_2(Y_j)\rfloor }}.
\end{equation}
\end{proposition}
\begin{proof}[Proof of \cref{prop:LoiYj}]
Under \cref{hyp:IM} there exist $N$ i.i.d. 
random variables $X_1,\dots, X_N$ uniform  in $(0,1)$ such that
$$
\{h(v), v\in \mathcal{M} \}\stackrel{\text{\emph{(d)}}}{=}\{X_1,\dots,X_N\}.
$$
Successive runs of Line \ref{l:w} return independent variables $W_k:=\{m X_k\}$ (where $\{\cdot\}$ stands for the fractional part) which are also independent and uniformly distributed in $(0,1)$.
% Argument loi géométrique
%, and $\rho(W_k)=- \lfloor \log_2(W_k)\rfloor$ is a positive geometric random variable with parameter $1/2$, \emph{i.e.} $\mathbb{P}(\rho(W_k)=k)=2^{-k}$ for all $k\geq 1$. %Variables $W_k$'s are also independent.

For $1\leq j\leq 2^b$  the random variable  $\lceil m X_k\rceil$ is uniform in $\{1,\dots,m\}$ and corresponds to the register associated to $X_k$. Set
$$
\mathcal{A}_j=\{X_k \in \mathbf{X}\text{ such that }\lceil m X_k\rceil =j\}.
$$
Then $A_j=|\mathcal{A}_j|$ denotes the number of hashed values lying in the $j$-th interval $\mathcal{A}_j$. By uniformity and independence of $W_k$'s,
$$
(A_1,\dots, A_m)\stackrel{\text{(d)}}{=} \text{Multinomial distribution with parameters }\left(N; (\tfrac1m,\dots \tfrac1m)\right).
$$
For convenience we identify a real $s\in [0,1]\equiv \{0,1\}^\mathbb{Z}$ with its proper binary decomposition, so that $\rho(y)=-\lfloor \log_2(y)\rfloor$ for $y<1$ and $\rho(1)=1$. Since $\rho(\cdot)$ is non-increasing, each $R[j]$ in Algorithm \ref{HLL} can be written as
$$
R[j]=\max\{\rho(W_k); W_k\in \mathcal{A}_j \}=\rho(Y_j),
$$
where% (recall that $\rho(s)=-\lfloor \log_2(s)\rfloor$, which is non-increasing),%$\{.\}$ stands for the fractional part and
\begin{equation}\label{eq:Yj}
Y_j=\min \{W_k,\ W_k\in \mathcal{A}_j \}.
\end{equation}
If $\mathcal{A}_j$ is empty we use in \cref{eq:Yj} the convention that $Y_j=1$. (This is consistent with our initialization $R[1]=\dots =R[j]=1=\rho(1)$.) Observe that conditionally to $|\mathcal{A}_j|=a_j>0$ then $Y_j$ is the minimum of $a_j$ uniform i.i.d. random variables and therefore has density $a_j(1-y)^{a_j -1}$. With these notation, $Z_N = \left(\sum_{j=1}^m 2^{-R[j] }\right)^{-1} \stackrel{\text{\emph{(d)}}}{=}  \left(\sum_{j=1}^m 2^{-\rho (Y_j) }\right)^{-1}$.
\end{proof}

In order to ease the analysis of the random variable $Z_N$ we may get rid of integer parts in \cref{eq:Z_proba} by using
$y/2\leq 2^{-\rho(y) }\leq y$ and write
\begin{equation}\label{eq:EncadrementZ}
 \frac{1}{\sum_{j=1}^m Y_j }  \leq Z_N\leq \frac{2}{\sum_{j=1}^m Y_j }.
\end{equation}
\begin{remark}\label{rem:m_1}
In the particular case $m=1$ we have with our notation $Z_N=2^{R[1]} \geq 1/Y_1$. Hence $\mathbb{E}[Z_N]\geq \int_{0}^1\frac{1}{y}N(1-y)^{N-1}dy=+\infty$, as mentioned earlier.
\end{remark}

Before going into the proof we give a quick heuristic explanation of why \texttt{HyperLogLog} returns a reasonable estimate for $N$, starting from \cref{eq:EncadrementZ}.

It is well known that for large $k$ the minimal value among $k$ i.i.d. uniform random variables in $(0,1)$ has a distribution close to that of $\frac{1}{k}\mathcal{E}$, where $\mathcal{E}$ is exponentially distributed with mean $1$. Here, 
as $Y_j$ is the minimal value among $|\mathcal{A}_j|\approx N2^{-b}$ uniform r.v. we should have that
$$
Y_j \stackrel{\text{(d)}}{\approx} \frac{1}{N2^{-b}}\mathcal{E}_j=\frac{m}{N}\mathcal{E}_j,
$$
where $\mathcal{E}_j$'s are exponential r.v. with mean $1$, and we expect
$$
Z_N \stackrel{\text{(d)}}{\approx}   \kappa \frac{1}{mN^{-1}(\mathcal{E}_1+\mathcal{E}_2+\dots +\mathcal{E}_m) }, 
$$
for some constant $1\leq \kappa \leq 2$.

Besides the intuition is that $\mathcal{E}_j$'s are almost independent if $N\gg m$. Since a sum of independent exponential random variables is distributed as a Gamma random variable we have 
 that  $\mathcal{E}_1+\mathcal{E}_2+\dots +\mathcal{E}_m \stackrel{\text{(d)}}{\approx} \Gamma(m,1)$. For large $m$, a $\Gamma(m,1)$-distributed random variable is fairly concentrated around $m$ and therefore the random variable $Z_N$ is concentrated around (some constant $\kappa$)$\times  N/m^2$. 
 This explains why $m^2Z_N$ is a good estimator of $N$, except that it may overestimate $N$ by a factor of $\kappa > 1$.

Beyond the ingenuity of the algorithm, a large portion of \cite{flajolet2007hyperloglog} is the  analytical \emph{tour-de-force} that makes it possible to correct the bias of $m^2Z_N$ by analyzing the constants
 involved  in the asymptotics  of $Z_N$.

\section{Proof of \cref{th:main}: deviation inequalities for the $Y_j$'s} 

Theorem \ref{th:main} follows from the following bounds for $\sum_j Y_j$, where $(Y_1,\dots,Y_m)$ is the random vector defined in \cref{prop:LoiYj}. %Put 
%$$
%\mathcal{I}(x)=x-1-\log(x)=\mathcal{J}(1/x),
%$$
Recall that $\mathcal{J}(x)=1/x-1+\log(x)$ was defined in \cref{eq:Rate}.
\begin{proposition}\label{prop:tail}
For every integers $N\geq 1$, $m\geq 1$ and every $0< c<1<C$ we have the following deviation inequalities:
 \begin{itemize}
\item \textbf{right-tail} 
\begin{equation}\label{eq:right}
\mathbb{P}\left(\sum_{j=1}^m Y_j\geq \frac{Cm^2}{N}\right)
%\leq \exp\left(-m\left(C-1+\log(C)\right)\right)
\leq \exp\left(-m\mathcal{J}(1/C)\right).
\end{equation}
\item \textbf{left-tail}
\begin{equation}\label{eq:left}
\mathbb{P}\left(\sum_{j=1}^m Y_j\leq \frac{cm^2}{N}\right)
%\exp\left( -m(c-1 +\log(1/c) \right)\times \exp( (1-c) \frac{m^2}{N} )
\leq \exp\left( -m\mathcal{J}(1/c) +  \delta(m,N,c)\right) ,
\end{equation}
where 
\begin{equation}\label{eq:delta}
 \delta(m,N,c)=c\frac{m^2}{N}+ m (1/c-1)e^{-N/m(1/c-1)-1}.
 \end{equation}
\end{itemize}
\end{proposition}

\begin{proof}[Proof of "\cref{prop:tail} implies \cref{th:main}"]
For the left-tail we apply \cref{eq:right} with $C=1/\mu >1$:
\begin{align*}
\mathbb{P}\left(m^2 Z_N \leq \mu N\right)
&\stackrel{\text{eq.\eqref{eq:EncadrementZ}}}{\leq}\mathbb{P}\left(\frac{m^2}{\sum_{j=1}^m Y_j} \leq \mu N\right)
=\mathbb{P}\left(\sum_{j=1}^m Y_j \geq \frac{m^2}{\mu N}\right)
&\stackrel{\text{eq.\eqref{eq:right}}}{\leq}  \exp\left(-m\mathcal{J}(\mu)\right).
\end{align*}
For the right-tail we apply \cref{eq:left} with $c=2/\lambda<1$:
\begin{align*}
\mathbb{P}\left(m^2 Z_N \geq \lambda N\right)\stackrel{\text{eq.\eqref{eq:EncadrementZ}}}{\leq} 
\mathbb{P}\left(\frac{2m^2}{\sum_{j=1}^m Y_j} \geq \lambda N\right)
&=\mathbb{P}\left(\sum_{j=1}^m Y_j \leq \frac{2m^2}{\lambda N}\right)\\
&\stackrel{\text{eq.\eqref{eq:left}}}{\leq} \exp\left(-m\mathcal{J}(\lambda/2) +\delta(m,N,2/\lambda)\right).
\end{align*}
Using the inequality $xe^{-x-1}\leq e^{-2}$ finally gives that $\delta(m,N,c)\leq (c+e^{-2})m^2/N$ as claimed in \cref{th:main}.
\end{proof}

\subsection{Proof of Proposition \ref{prop:tail}: right-tail}
The purpose of this section is to prove the right-tail deviation inequality given in \cref{eq:right}. To do so we will make the intuition of \cref{sec:heuristic} rigorous by showing a comparison lemma for the $Y_j$'s.

For two $m$-dimensional random vectors $\mathbf{X}=(X_1,\dots X_m)$ and $\mathbf{X}'=(X'_1,\dots X'_m)$ we write 
$ \mathbf{X} \preccurlyeq_{\text{u.o.}} \mathbf{X}' $ to say that $\mathbf{X}$ is smaller than $\mathbf{X}'$ in the \emph{upper orthant order}, \emph{i.e.}  $\mathbb{P}(X_1 > t_1,X_2>t_2,\dots ,X_m>t_m)\leq \mathbb{P}(X_1' > t_1,X_2'>t_2,\dots ,X_m'>t_m)$ for every reals $t_1,\dots,t_m$ (see \emph{e.g.} \cite[Chap.6]{shaked2007stochastic} for a review of multivariate stochastic orders).

%Below we use the usual notation $X\preccurlyeq X'$ when a random variable $X$ is \emph{stochastically dominated} by $X'$, meaning that $\mathbb{P}(X> t) \leq \mathbb{P}(X'> t)$ for every real $t$.
\begin{lemma}\label{lem:dominationY}
For $N\geq 1$ and  $1\leq j\leq m$ let $(Y_1,\dots ,Y_m)$ be as in \cref{prop:LoiYj}. Let also $\mathcal{E}_1,\dots, \mathcal{E}_m$ be  independent exponential random variables with mean $1$. Then
$$
(Y_1,\dots ,Y_m)\preccurlyeq_{\text{u.o.}} \frac{m}{N}(\mathcal{E}_1,\dots, \mathcal{E}_m).
$$
%Furthermore the family $\{Y_1,\dots , Y_m\}$ is negatively associated.
\end{lemma}
%We first observe a direct consequence of \cref{hyp:IM}. Since the minimum of $k$ uniform random variables in $(0,1)$ has density $k(1-t)^{k-1}\mathrm{dt}$ then we can compute the joint distribution of $Y_j$'s.
%\begin{claim}\label{claim:density_Yj}
%Let $k_1,\dots ,k_m$ be positive integers such that $\sum_j k_j=N$. Conditionally to $A_1=k_1,\dots , A_m=k_m$ then $(Y_1,\dots,Y_m)$ has joint density
%$$
%k_1(1-t_1)^{k_1-1}\mathbf{1}_{t_1\in (0,1)}\times k_2(1-t_2)^{k_2-1}\mathbf{1}_{t_2\in (0,1)}\ \times \dots \times k_m(1-t_m)^{k_m-1}\mathbf{1}_{t_m\in (0,1)}\mathrm{dt_1\dots dt_m}.
%$$
%(In particular $Y_j$'s are independent conditioned to $A_j$'s.)
%\end{claim}

\begin{proof}[Proof of \cref{lem:dominationY}]
Let $t_1,\dots ,t_m$ be reals in $(0,1)$. Recall from \cref{prop:LoiYj} that conditionally to $A_1=a_1,\dots, A_m=a_m$ the variables $Y_j$'s are independent and if $a_j>0$ then each $Y_j$ has density $a_j(1-y)^{a_j-1}$. Hence for every 
$a_j>0$ one has $\mathbb{P}(Y_j > t_j | A_j=a_j)=(1-t_j)^{a_j}$, and this inequality also holds if $a_j=0$ thanks to the initialization $Y_j\leftarrow 1$.

Hence by the law of total expectation
%The event $\{Y_1 > t_1;Y_2>t_2;\dots Y_m>t_m\}$ occurs exactly when every hashed value lies outside $\cup_j (jm,jm+t_j)$. Then
\begin{align*}
\mathbb{P}(Y_1 > t_1,Y_2>t_2,\dots ,Y_m>t_m)&=
\mathbb{E}\left[\mathbb{P}(Y_1 > t_1,Y_2>t_2,\dots ,Y_m>t_m\ |\ A_1,\dots,A_m)\right]\\
&=\mathbb{E}\left[(1-t_1)^{A_1}\times \dots \times (1-t_m)^{A_m}\right]\\
&=\sum_{a_1+\dots +a_m=N}\binom{N}{a_1 a_2 \ \dots a_m}\frac{1}{m^N}(1-t_1)^{a_1}\times \dots \times(1-t_m)^{a_m}\\
&=\left(1-\frac{1}{m}t_1-\dots -\frac{1}{m}t_m\right)^N\\
&\leq \exp\left(-\frac{N}{m}\sum_j t_j\right) = \prod_{j=1}^m \mathbb{P}\left(\frac{m}{N}\mathcal{E}_j > t_j\right).
\end{align*}
 If one of the $t_j$'s is greater or equal to $1$ then the left-hand side in the above display is zero, so the last inequality is in fact true for every  reals $t_1,\dots ,t_m$ and the lemma is proved.
\end{proof}

\begin{proof}[Proof of  \cref{prop:tail}, \cref{eq:right}]
We want to apply the Chernov method  (see \emph{e.g.} \cite[Chap.4]{mitzenmacher2017probability}).
 Let $C>1$ be arbitrary and let $\lambda>0$ to be chosen later and satisfying  $0<\lambda m/N<1$,
$$
\begin{array}{r c l l}
\mathbb{P}\left(\sum_{j=1}^m Y_j > \frac{Cm^2}{N}\right)%&\leq& \mathbb{P}\left(\sum_{j=1}^m \mathcal{E}_j > Cm\right) \\
&\leq& \mathbb{P}\left(e^{\lambda \Sigma_j Y_j} > e^{\lambda Cm^2/N}\right)& \\
&\leq& e^{-\lambda Cm^2/N}\mathbb{E}[e^{\lambda \Sigma_j Y_j}]. &%\text{(Markov's inequality).}
\end{array}
$$
Using \cite[Th.6.G.1]{shaked2007stochastic}, the upper orthant domination given by \cref{lem:dominationY} implies a comparison for the Laplace transforms. Namely, for every $\lambda >0$, 
$$
\mathbb{E}\left[\exp\left( \lambda \sum_{j=1}^m Y_j\right)\right] \leq \mathbb{E}\left[\exp\left( \lambda \frac{m}{N}\sum_{j=1}^m \mathcal{E}_j\right)\right].%=\mathbb{E}\left[\exp\left( \lambda \frac{m}{N} \mathcal{E}_1\right)\right]^m
$$
Therefore
$$
\begin{array}{r c l l}
\mathbb{P}\left(\sum_{j=1}^m Y_j > \frac{Cm^2}{N}\right)%&\leq& \mathbb{P}\left(\sum_{j=1}^m \mathcal{E}_j > Cm\right) \\
&\leq& e^{-\lambda Cm^2/N}\mathbb{E}[\exp\left(\lambda \frac{m}{N} \Sigma_j \mathcal{E}_j\right)] & \\
&\leq& e^{-\lambda Cm^2/N}\mathbb{E}[\exp\left(\lambda \frac{m}{N}  \mathcal{E}_1\right)]^m & \\
&\leq& \exp\left( -\lambda Cm^2/N -m\log(1-\lambda m/N)\right), &
\end{array}
$$
where we used the identity $\mathbb{E}[e^{u \mathcal{E}_1}]=1/(1-u)$ with $u=\lambda m/N$ (this is why we need $0<\lambda m/N<1$). Choosing now $\lambda=\frac{N}{m}(1-1/C)$, which satisfies $\lambda m/N<1$, we obtain
$$
\mathbb{P}\left(\sum_{j=1}^m Y_j > \frac{Cm^2}{N}\right)\leq \exp\left(-m\left(C-1)+m\log(C)\right)\right),
$$
which is  \cref{eq:right}.
\end{proof}

\subsection{Proof of \cref{prop:tail}: left-tail}
%For $1\leq j\leq m$ let $A_j=|\mathcal{A}_j|$ denote the number of hashed values lying in the $j$-th interval. On the (very unlikely) event $A_j=0$ then $Y_j=1$ and this value will annoy us. For this reason, we will show the lemma for a variant of the $Y_j$. %Before embarking in the proof, we observe that this is enough to prove \cref{prop:tail} with a small variant.
%For each $1\leq j\leq m$ put 
%$$
%\mathcal{Y}_j=
%\begin{cases}
%Y_j &\text{ if }A_j>0,\\
%0 &\text{ if }A_j=0
%\end{cases}.
%$$
%Since $\mathcal{Y}_j\leq Y_j$ we have $\mathbb{P}\left(\sum_{j=1}^m Y_j\leq \frac{c}{N}\right)\leq \mathbb{P}\left(\sum_{j=1}^m \mathcal{Y}_j\leq \frac{c}{N}\right)$ it is enough to prove a deviation inequality for the $\mathcal{Y}_j$'s.

The idea behind the proof of the left-tail is to circumvent the problem of the dependence of $Y_j$'s using a concavity argument. We begin with a simple estimation of the Laplace transform of $Y_j$.
\begin{lemma}\label{lem:majo_laplace}
For every $\lambda >1$ and $k\geq 0$,
$$
\mathbb{E}[e^{-\lambda Y_j}\ |\ A_j=k]%=k\int_{0}^1\exp(-\lambda u/m)(1-u)^{k-1}\mathrm{du} 
%\leq \exp\left(-\frac{\lambda}{k+1}+\frac{\lambda^2}{(k+1)^2}\right).%\leq \exp\left(-\frac{\lambda/m}{\lambda/m +k}\right).
\leq \frac{k+e^{-\lambda}(\lambda-1)}{k+\lambda-1} =:F(k,\lambda).
$$
\end{lemma}

\begin{proof}[Proof of  \cref{lem:majo_laplace}]
Recall that if $A_j=0$ then $Y_j=1$ by our choice of initialization. In this case both sides of the inequality are $e^{-\lambda}$ so the the inequality holds for $k=0$. We now assume $k\geq 1$.  %Let $\mathcal{B}_k$ be a random variable having this distribution, \emph{i.e.} having the Beta$(1,k)$ distribution. 
By \cref{prop:LoiYj},
\begin{align*}
\mathbb{E}[e^{-\lambda Y_j}\ |\ A_j=k]&= \int_0^1 e^{-\lambda y} k(1-y)^{k-1}\mathrm{dy} \\
&\leq  k \int_0^{+\infty} e^{-(k+\lambda-1)y} \mathrm{dy} \qquad \text{(using $1-y\leq \exp(-y)$)}\\
&=  \frac{k}{k+\lambda-1} \leq   \frac{k+e^{-\lambda}(\lambda-1)}{k+\lambda-1}.
\end{align*}

\end{proof}

\begin{proof}[Proof of \cref{prop:tail}, \cref{eq:left}] By  \cref{prop:LoiYj}  again we have that, conditionally on $A_1,\dots,A_m$, the $Y_j$'s are independent. Therefore
$$
\begin{array}{r c l l}
\mathbb{E}[e^{-\lambda \Sigma_j Y_j}]&=&\mathbb{E}\left[\mathbb{E}[e^{-\lambda \Sigma_j Y_j}\ |\ A_1,A_2,\dots, A_m]\right] & \\
&=&\mathbb{E}\left[\prod_{j=1}^m \mathbb{E}\left[e^{-\lambda Y_j}\ |\ A_j\right]\right]& \text{(cond. independence)}\\
&\leq&\mathbb{E}\left[\prod_{j=1}^m  F(A_j,\lambda)\right].& \text{(\cref{lem:majo_laplace})}
\end{array}
$$
%where we put
%$$
%F(x,\lambda)= \frac{1}{1+\frac{\lambda}{x+1}}.
%$$
%Using now the inequality $\frac{1}{1+u}\leq e^{u/(1+u)}$ with $u=\lambda/(A_j+1)$ we get
%\begin{align*}
%\mathbb{E}\left[\prod_{j=1}^m  F(A_j,\lambda)\right]
%&=\mathbb{E}\left[\exp\left(-\sum_{j=1}^m \frac{\lambda}{A_j+1+\lambda}\right)\right]
%\end{align*}
Consider now for any fixed $\lambda>1$ the function $f_\lambda:x\mapsto  \log(F(x,\lambda))$. We have that
$$
f_\lambda''(x) = \frac{1}{(x+\lambda-1)^2}-\frac{1}{(x+e^{-\lambda}(\lambda -1))^2}.
$$
Hence $f_\lambda$ is concave over $[0,+\infty)$  for $\lambda>1$ and this implies  
$\frac{1}{m}\sum_{j=1}^m f_\lambda(A_j)\leq  f_\lambda\left( \frac{1}{m}\sum_{j=1}^m A_j \right)=f_\lambda(N/m)$. 
 Therefore, for every $\lambda >1$,
\begin{equation}\label{eq:logconcave}
\mathbb{E}[e^{-\lambda \Sigma_j Y_j}]\leq 
\mathbb{E}\big[\prod_{j=1}^m  F(A_j,\lambda)\big]
= \mathbb{E}\big[\exp\big(\sum_{j=1}^m f_\lambda(A_j)\big)\big] 
\leq  \exp\left(mf_\lambda(N/m)\right).
%= \left( F(N/m,\lambda)\right)^m.
\end{equation}
%Hence
%\begin{array}{r c l l}
%\mathbb{E}[e^{-\lambda \Sigma_j \mathcal{Y}_j}]&\leq & \mathbb{E}[\prod_{j=1}^m  \frac{A_j}{\lambda/m +A_j}] & \\
%&=& \mathbb{E}[\exp\left(\sum_{j=1}^m f(A_j)\right)] &\\
%&\leq&  \exp\left(mf(N/m)\right)& \\
%&=& \left( \frac{N}{N+\lambda}\right)^m.&
%\end{array}
As for the proof of the right-tail inequality we conclude with the Chernov method. Let $c<1$,
$$
\begin{array}{r c l l}
\mathbb{P}\left(\sum_{j=1}^m Y_j\leq \frac{cm^2}{N}\right)
%&\leq & \mathbb{P}\left(e^{-\lambda \Sigma_j Y_j}\geq e^{-\lambda cm^2/N}\right) &\\ 
&\leq & e^{\lambda cm^2/N}\mathbb{E}[e^{-\lambda \Sigma_j  Y_j}] &\\ 
&=& \exp\left( \lambda cm^2/N + mf_\lambda(N/m)\right)  & \\%\text{(using eq.\eqref{eq:logconcave})}\\
&=& \exp\left( \lambda c\frac{m^2}{N} +m\log(N/m+e^{-\lambda}(\lambda -1)) - m\log \left(N/m +\lambda -1\right)\right).& \\
\end{array}
$$
If we now choose $\lambda=N/m\times (1/c-1)+1$, which is $>1$ as required, then we obtain
$$
\begin{array}{r c l l}
\mathbb{P}\left(\sum_{j=1}^m Y_j\leq \frac{cm^2}{N}\right)
&\leq& \exp\bigg( m(1-c) + \delta_1(m,N,c) & \\
& &+ m\log(N/m + \delta_2(m,N,c)) -m\log(N/m)-m\log(1/c) \bigg), & %\\
%&=&\exp\left( -m\mathcal{I}(c) + \delta(m,N,c)\right) &
\end{array}
$$
where we put
\begin{align*}
\delta_1(m,N,c)&=c \frac{m^2}{N},\\
\delta_2(m,N,c)&=e^{-\lambda}(\lambda -1)=\exp(-N/m(1/c-1)-1)\times N/m\times (1/c-1).
\end{align*}
Using finally the mean-value inequality $\log(u+\delta)\leq \log(u)+\frac{\delta}{u}$ we obtain
\begin{align*}
\mathbb{P}\left(\sum_{j=1}^m Y_j\leq \frac{cm^2}{N}\right)&\leq \exp\left( -m\left(-1+c+\log(1/c)\right) + \delta_1(m,N,c)+\frac{m}{N/m}\delta_2(m,N,c)\right)\\
%&\leq  \exp\left( -m\left(-1+c+\log(1/c)\right) + \delta(m,N,c)\right),\\
&\leq  \exp\left( -m\mathcal{J}(1/c) + \delta(m,N,c)\right),
\end{align*}
where $\delta(m,N,c)$ is as stated in \cref{prop:tail}.
\end{proof}

%%%%%%%%%%%%%%%%%
\section{Final comments}

\subsection{One-sided exact confidence interval for $N$}
One practical motivation to the original algorithm is the detection of attacks on large networks (see the discussion in \cite{giroire2009order}). Since certain attacks generate an abnormal number of connections with identical addresses, it is crucial to have good one-sided confidence intervals for $N$. This makes it possible to detect an excessively high $N$ value with very few false positives.

Plugging the simple estimate $\mathcal{J}(1/x)\geq \frac{1}{4}(x-1)^2$ valid for all $0<x<2$ into \cref{eq:main-left} gives that for all $0<\eps<1$ 
%$$
%\mathbb{P}\left(N > \left(1+2\sqrt{\frac{1}{m}\log(1/\eps)}\right) \times m^2 Z_N \right) \leq \eps.
%$$
%In other words,
$$
\left[\ 0\ ,\ \left(1+2\sqrt{\frac{1}{m}\log(1/\eps)}\right) \times m^2 Z_N \right]
$$
contains the actual value of $N$ with probability larger than $1-\eps$. In other words this is a non-asymptotic one-sided confidence interval for $N$ with confidence $\geq 1-\eps$.

To give an idea for $\eps=0.001$ and $m=128$, this guarantees that  $N$ is less than $1.465\times m^2 Z_N$ with probability larger than $0.999$.

\subsection{How to improve the right-tail inequality: $2^{-\rho(Y_j)}$ instead of $Y_j$ }\label{sec:vraie_laplace}
There is room for improvement in \cref{th:main}: it is likely that the left-hand side in the right-tail inequality \cref{eq:main-right} decays exponentially $\lambda > \lambda_0$ with $\lambda_0 < 2$.
A way to improve this is to keep $2^{-\rho(Y_j)}$ instead of $Y_j$ in the analysis and replace the upper bound of $\mathbb{E}[e^{-\lambda Y_j}\ |\ A_j=k]$ in \cref{lem:majo_laplace}  by a sharp upper bound for
\begin{align*}
\mathcal{L}(k,\lambda):=\mathbb{E}[e^{-\lambda 2^{-\rho(Y_j)}}\ |\ A_j=k]&=\sum_{n=1}^{+\infty} e^{-\lambda 2^{-n}} \mathbb{P}(2^{-n}\leq Y_j \leq 2^{-(n-1)})\\
%&=\sum_{n=1}^{+\infty} e^{-\lambda 2^{-n}} \int_{2^{-n}}^{2^{-n+1}} k(1-y)^{k-1}dy\\
&=\sum_{n=1}^{+\infty} e^{-\lambda 2^{-n}} \left((1-2^{-n})^{k}- (1-2^{-(n-1)})^{k} \right).
\end{align*}

Lemma 1 in \cite{flajolet2007hyperloglog} provides clues for getting an upper bound $\mathcal{L}(k,\lambda)\leq F(k,\lambda)$ for some $F$ but it does not seem easy to obtain a sharp, simple, and non-asymptotic upper bound. Furthermore, we need to find a function $F$ which is log-concave in $k$ in order to apply \cref{eq:logconcave}. We did not pursue the calculations further in that direction.

\subsection{Applications to \textsc{MinCount}}\label{sec:MinCount}
We briefly discuss an application of \cref{prop:tail}  to the analysis of \texttt{MinCount} proposed by Giroire \cite{giroire2009order}. This algorithm can be seen as a continuous cousin of \texttt{HyperLogLog}. 
Using the same notation as in \cref{sec:heuristic},  \texttt{MinCount} returns a simple function of $Y^{(k)}_1,\dots,Y^{(k)}_m$, where
$$
Y^{(k)}_j = \text{$k$-th smallest hashed value among }\left\{X_\ell,\ X_{\ell}\in\left((j-1)\tfrac{1}{m},j\tfrac{1}{m}\right)\right\}.
$$
Several estimators are proposed in \cite{giroire2009order}, information-theoretic arguments \cite{chassainggerin} later suggested to return\footnote{Note that we have slightly modified the renormalization compared to \cite{chassainggerin}, in order to ensure that the notation is consistent with the present article. Hence the factor $km^2$ instead of $km-1$.}
$$
\tilde{N}:=\frac{km^2}{\sum_{j=1}^{m}Y^{(k)}_j}.
$$
Observe that the case $k=1$ is very similar to \texttt{HyperLogLog} (see Algorithm \ref{algo:MinCount}).
\begin{algorithm}
\caption{MinCount, case $k=1$}\label{algo:MinCount}
\begin{algorithmic}[1]
 \State  \textbf{input:} multiset $\mathcal{M}$ of items from domain $\mathcal{D}$
 \State  \textbf{parameters:} integer $b\geq 1$; $m \gets 2^b$
\State \textbf{initialize} a collection of $m$ registers $Y_1=\dots =Y_m=1$
\For{$v\in \mathcal{M}$} 
\State $x \gets h(v)$ and write $x=0.x_1x_2x_3\dots$ \hfill // hashing
\State $j \gets 1+ \langle x_1x_2\dots x_b\rangle_2$ 	\hfill // register determined by the first $b$ bits of $x$
\State $w \gets 0.x_{b+1}x_{b+2}\dots $ \hfill  // extracting a fresh binary string out of $x$
\State $Y_j \gets \mathrm{min}\{Y_j,w\}$ \hfill // updating the minimum in the $j$-th register 
\EndFor
\State  \textbf{return} $\tilde{N}:=\frac{m^2}{\sum_{j=1}^{m}Y_j}$ 
\end{algorithmic}
\end{algorithm}

If we apply \cref{prop:tail} to Algorithm \ref{algo:MinCount} we obtain sharp deviation estimates:
\begin{proposition}[Exponential deviations for \texttt{MinCount} (case $k=1$)]\label{prop:MinCount}\ \\
Assume \cref{hyp:IM}. For every $N\geq 1$ the output of  Algorithm \ref{algo:MinCount} satisfies the following inequalities. 
\begin{itemize}
\item {\bf Left-tail.} For all $\mu \leq 1$
$$
\mathbb{P}\left(\tilde{N} \leq \mu N\right)\leq  \exp\left(-m\mathcal{J}(\mu) \right).
$$
\item {\bf Right-tail.} For all $\lambda \geq 1$
$$
\mathbb{P}\left(\tilde{N} \geq \lambda N\right)\leq \exp\left(-m\mathcal{J}(\lambda) +\delta(m,N,1/\lambda) \right),
$$
where $ \delta(m,N,c)$ is  given by \cref{eq:delta}.
\end{itemize}
\end{proposition}
When $k=3$, Giroire shows that the estimators based on $Y^{(k)}_1,\dots,Y^{(k)}_m$ exhibit much better behavior in terms of both expectation and variance. %Unfortunately, a simple analogue of \cref{prop:MinCount} for $k\neq 1$ would require a substantial amount of work: there is no stochastic domination as elegant as that of \cref{lem:dominationY}. %
\cref{prop:MinCount} however shows that even if the case $k=1$ \texttt{MinCount} provides a fairly good estimation of $N$.
For example, for $m=128$ and $N=10^9$ this bounds guarantee
$$
\mathbb{P}\left( \tilde{N}  \leq 0.5 N\right)<  9\times 10^{-18},\qquad
\mathbb{P}\left( \tilde{N}  \geq 1.5 N\right) < 10^{-4}.
$$

\bibliography{BiblioHLL}
\end{document}